\documentclass[lettersize,journal]{IEEEtran}
\usepackage{amsmath,amsfonts}
\usepackage{setspace}
\usepackage{algorithmic}
\usepackage{algorithm}
\usepackage{array}
\usepackage{textcomp}
\usepackage{stfloats}
\usepackage{url}
\usepackage{verbatim}
\usepackage{graphicx}
\usepackage{cite}
\usepackage{caption}
\usepackage{subcaption}
\usepackage{amssymb}
\usepackage[colorlinks, linkcolor=blue, anchorcolor=blue, citecolor=green]{hyperref}
\usepackage{soul}       
\usepackage{xcolor}     
\usepackage{stfloats}
\newtheorem{proposition}{Proposition}
\begin{document}

\title{\LARGE Multi-AP Cooperative Beamforming for Cell-Free ISAC Networks:
Balancing Communication SINR and Sensing SCNR}

\author{Jijin~Guo,
        Lixin Li,~\IEEEmembership{Member,~IEEE,}
        Yufeng Zheng,~\IEEEmembership{Student~Member,~IEEE,}
        Dongwei Zhao, \\
        Wensheng Lin,~\IEEEmembership{Member,~IEEE,} 
        and~Zhu Han,~\IEEEmembership{Fellow,~IEEE}%
        	
        \thanks{Corresponding authors: Lixin Li, Wensheng Lin.

	This work was supported in part by National Natural Science Foundation of China under Grants 62571450 and 62101450, in part by Key Research and Development Program of Shaanxi under Grant 2025CY-YBXM-043, in part by Shanghai Academy of Spaceflight Technology under Grant SAST2025-037, in part by the Open Fund of Intelligent Control Laboratory, in part by the Open Fund of Key Laboratory of Radio Spectrum Testing Technology (The State Radio\_monitoring\_center Testing Center), Ministry of Industry and Information Technology, the National Science Foundation (NSF) under Grant ECCS-2302469; the Japan Science and Technology Agency (JST) Adopting Sustainable Partnerships for Innovative Research Ecosystem (ASPIRE) under Grant JPMJAP2326.}
	\thanks{Jijin~Guo, Lixin Li, Yufeng Zheng, Dongwei Zhao, and Wensheng Lin are with the
                School of Electronics and Information, Northwestern Polytechnical University,
                Xi’an, Shaanxi 710129, China (e-mail: guojijin@mail.nwpu.edu.cn; lilixin@nwpu.edu.cn; zhengyufeng@mail.nwpu.edu.cn; zhaodongwei@mail.nwpu.edu.cn; linwest@nwpu.edu.cn).} 
		\thanks{Zhu Han is with the Department of Electrical and Computer Engineering at the University of Houston, Houston, TX 77004 USA (e-mail: hanzhu22@gmail.com).}
}



\maketitle

\begin{abstract}
Cell-free integrated sensing and communication (ISAC) systems are facing the resource allocation challenges due to the deployment of access points (APs) and conflicting beamforming requirements between the communication and sensing functions. Unlike traditional ISAC architectures, the geographic distribution of APs introduces coordination complexity and resource-sharing conflicts that existing single-objective methods cannot adequately address. To address this challenge, we formulate an optimization problem for multi-AP cooperative beamforming that maximizes the sensing signal-to-clutter-plus-noise ratio (SCNR) under the communication rate constraints. The non-convex quadratically constrained quadratic program is transformed 
into a tractable convex semidefinite program via semidefinite relaxation, enabling efficient polynomial-time solutions and overcoming the local convergence limitations of traditional alternating optimization approaches. Simulation results demonstrate that the proposed approach achieves superior performance in both communication signal-to-interference-plus-noise ratio (SINR) and SCNR compared to existing schemes, confirming its effectiveness for balancing dual-functional objectives.
\end{abstract}

\begin{IEEEkeywords}
Cell-free networks, integrated sensing and communication, cooperative beamforming, semidefinite relaxation.
\end{IEEEkeywords}

\section{Introduction}
\IEEEPARstart{W}{ith} the rapid evolution of the sixth-generation (6G) networks, integrated sensing and communication (ISAC) has emerged as a key enabling technology that unifies communication and sensing functionalities within a single network architecture. By sharing hardware infrastructure and wireless spectrum resources, ISAC not only reduces network deployment and operational costs, but also facilitates joint resource utilization
between communication and sensing functions, thereby providing strong technical support for a wide range of emerging
applications such as security monitoring, healthcare, and intelligent transportation~\cite{b1, b2, b19}. 

Although existing studies have demonstrated the performance benefits of ISAC in single-base-station scenarios, research on practical multi-base-station cooperative networks remains relatively limited. Multi-base-station cooperation introduces complex challenges in resource allocation, interference management, beamforming design, and collaborative sensing. Addressing these challenges effectively is crucial for improving network performance and enhancing spectrum efficiency. Meanwhile, cell-free massive multiple-input multiple-output (CF-mMIMO), which eliminates traditional cellular boundaries and enables fully cooperative service among distributed access points (APs), has been recognized as an ideal architecture for ISAC deployment~\cite{b3}.

Several studies have focused on enhancing communication and sensing performance in ISAC systems. For instance, \cite{b4} utilizes signal-to-interference-plus-noise ratio (SINR) as a unified metric to address interference issues under both perfect and imperfect channel state information (CSI) conditions. \cite{b6} proposes a power allocation algorithm to harmonize communication and sensing, thereby boosting target detection capabilities. \cite{b7} and \cite{b20} investigate cooperative beamforming and optimal transport-based resource allocation for unmanned aerial vehicle (UAV) scenarios, respectively, while \cite{b8} presents a robust beamforming technique to tackle non-convexity challenges. \cite{b17} investigates joint sensing and bi-directional communication within a dynamic time-division duplexing (TDD) framework, focusing on target detection and radar cross section (RCS) estimation. Additionally, \cite{b18} addresses cooperative sensing by optimizing the Cramér-Rao Bound (CRB) to improve parameter estimation precision in the presence of synchronization errors. Despite these advancements, many methods still face significant challenges in addressing severe resource conflicts
and collaborative sensing difficulties within the distributed AP
networks, particularly in balancing the communication SINR and
sensing signal-to-clutter-plus-noise ratio (SCNR), where substantial performance gaps remain.

Motivated by the identified research gaps, this letter proposes a multi-AP cooperative beamforming framework for cell-free ISAC (CF-ISAC) networks to maximize sensing SCNR under communication rate constraints. We explicitly address three fundamental challenges: per-AP power constraints, inherent resource conflicts, and multi-static sensing geometry. The resulting non-convex problem is solved via semidefinite relaxation (SDR) with guaranteed global convergence. Crucially, by constructing a strictly rank-one solution to prove the tightness of the SDR, our approach achieves the global near-optimal balance between communication SINR and sensing SCNR, effectively eliminating the performance losses inherent in existing relaxation or heuristic methods. Simulation results demonstrate that the proposed scheme achieves near-optimal performance and robust spatial resilience. These gains scale effectively with antenna and power configurations, highlighting the viability of this approach for 6G networks.
\section{System Model}
As illustrated in Fig.~\ref{fig:1}, the CF-ISAC system consists of $M$ APs. Among them, $M_t$ APs are designated as transmit APs, denoted by the set $\mathcal{M}_t$ with $\left| \mathcal{M}_t \right| = M_t$, while the remaining $M_r$ APs are configured as receive APs, denoted by the set $\mathcal{M}_r$ with $\left| \mathcal{M}_r \right| = M_r$. Specifically, $\mathcal{M}_t$ is selected via user-centric clustering based on the strongest large-scale channel gains to users, while $\mathcal{M}_r$ consists of APs with the minimum path loss to the targets to ensure optimal echo reception. Each AP is equipped with $N$ antennas and is capable of simultaneously serving $K$ single-antenna communication users and sensing $Q$ targets. All APs employ uniform linear arrays (ULAs) with half-wavelength spacing for high-resolution beamforming. A CPU coordinates all APs via fronthaul links, ensuring perfect time and frequency synchronization to enable coherent cooperative communication and distributed sensing. 
Each AP utilizes a fully digital architecture, where every antenna is tied to an independent radio-frequency (RF) chain to support simultaneous multi-stream communication and sensing signaling.

The signal transmitted by the $m$-th AP for the $\ell$-th symbol is given by
\begin{figure}
    \centering
    \includegraphics[width=0.8\linewidth]{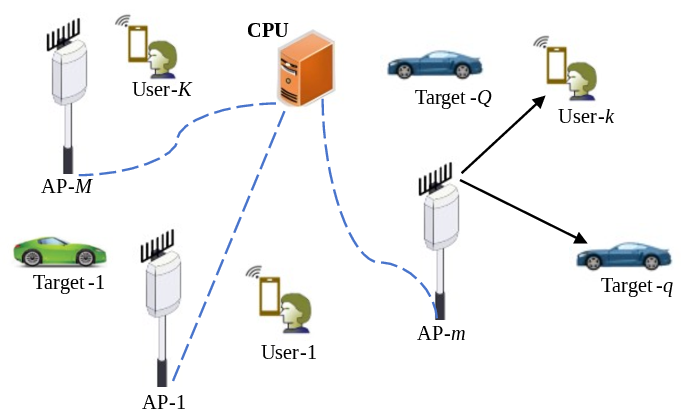}
    \caption{A CF-ISAC system model.}
    \label{fig:1}
\end{figure}
\begin{equation}\mathbf{x}_m[\ell]=\sum_{k\in\mathcal{K}}\mathbf{v}_{mk}x_k[\ell],\end{equation}
where $x_k[\ell] \in \mathbb{C}$ denotes the $\ell$-th symbol of the $k$-th stream, and $\mathbf{v}_{mk} \in \mathbb{C}^{N \times 1}$ is the beamforming vector applied by the $m$-th AP to this stream. We consider an orthogonal frequency-division multiplexing (OFDM)-based system where the cyclic prefix (CP) is sufficiently long to accommodate all propagation delays, including both user signals and target echoes. This allows for per-subcarrier processing, and our analysis thus focuses on a representative single subcarrier. The symbols are assumed to have unit average energy, $\mathbb{E}\left[ |x_k|^2 \right] = 1$. The beamforming vectors are subject to a total power constraint $P_m$ at each AP, expressed as $\mathbb{E}\left[ \| \mathbf{x}_m[\ell] \|^2 \right] = \sum_{k \in \mathcal{K}} \|\mathbf{v}_{mk}\|^2 \leq P_m$. Furthermore, by stacking the beamforming vectors from all APs for the $k$-th stream, we define the aggregate beamforming vector as $\mathbf{v}_k = \left[ \mathbf{v}_{1k}^\mathrm{T}, \ldots, \mathbf{v}_{M_t k}^\mathrm{T} \right]^\mathrm{T} \in \mathbb{C}^{M_t N}$.
\subsection{Communication Model}\label{AA}

The communication environment is modeled using a block flat-fading channel~\cite{b12}, where the channel coefficients are assumed to remain constant within each coherence block. Within each block, the channel vector $\mathbf{h}_{mk} \in \mathbb{C}^{N \times 1}$ characterizes the communication link from the $m$-th AP to the $k$-th user. Under TDD operation, channel reciprocity is assumed to be perfectly calibrated, allowing the APs to leverage uplink-acquired CSI for downlink beamforming. Specifically, the communication channel vector $\mathbf{h}_{mk}$ is mathematically defined as
\begin{equation}
\mathbf{h}_{mk} = \zeta_{\mathbf{h}_{mk}}^{1/2} \widetilde{\mathbf{h}}_{mk}, \quad \forall k, m,
\end{equation}
where $\zeta_{\mathbf{h}_{mk}}$ represents the large-scale path-loss factor, which remains constant over consecutive coherence intervals, and $\widetilde{\mathbf{h}}_{mk} \sim \mathcal{CN}(\mathbf{0}, \mathbf{I}_N)$ denotes the small-scale fading component, exhibiting typical Rayleigh fading characteristics. Furthermore, by stacking the channels between user $k$ and all APs, we construct the aggregate channel vector $\mathbf{h}_k \in \mathbb{C}^{M_t N \times 1}$. The received signal at user $k$ can be expressed as
\begin{align}
y_k^{(\mathrm{c})}[\ell] &= \sum_{m \in \mathcal{M}_t} \mathbf{h}_{mk}^\mathrm{H} \mathbf{x}_{m}[\ell] + n_k[\ell] \nonumber\\
&= {\sum_{m \in \mathcal{M}_t} \mathbf{h}_{mk}^\mathrm{H} \mathbf{v}_{mk} x_{k}[\ell]} \nonumber\\
&\quad + {\sum_{k' \in \mathcal{K} \setminus \{k\}} \sum_{m \in \mathcal{M}_t} \mathbf{h}_{mk}^\mathrm{H} \mathbf{v}_{mk'} x_{k'}[\ell]}
+ {n_k[\ell]},
\end{align}
where $n_k[\ell] \sim \mathcal{CN}(0, \sigma_k^2)$ represents the receiver noise at user $k$. The SINR for user $k$ is given by
\begin{equation}
\gamma_k^{(\mathrm{c})} = \frac{\left|\mathbf{h}_k^\mathrm{H} \mathbf{v}_k\right|^2}{\sum_{k' \in \mathcal{K} \setminus \{k\}} \left|\mathbf{h}_k^\mathrm{H} \mathbf{v}_{k'}\right|^2 + \sigma_k^2}.
\label{eq:4}
\end{equation}

The communication rate for user $k$ is then expressed as
\begin{equation}
R_k = \log_2(1 + \gamma_k^{(\mathrm{c})}).
\end{equation}
\subsection{Sensing Model}\label{AA}
In sensing task modeling, receiving APs use prior knowledge of transmitted signals to perform environmental sensing. Assuming the direct paths are removed via channel estimation, we adopt the point reflector model~\cite{b10}. The channel from the transmitting AP $m_t$ to the receiving AP $m_r$ via the $q$-th target reflection can be expressed as
\begin{equation}
\mathbf{G}_{m_t, q, m_r} = g_{m_t, q, m_r} \mathbf{a}(\phi_{m_r, q}) \mathbf{a}^\mathrm{H}(\phi_{m_t, q}),
\end{equation}
where $g_{m_t, q, m_r} \sim \mathcal{CN}(0, \zeta_{m_t, q, m_r}^2)$ is the combined sensing
channel gain, which includes the effects due to the path loss and RCS of the $q$-th target. $\phi_{m_t, q}$ and $\phi_{m_r, q}$ denote the angle of departure (AoD) and angle of arrival (AoA), respectively. The array steering vector is defined as $\mathbf{a}(\phi) \triangleq \frac{1}{\sqrt{N}} [1, e^{\mathrm{j}\pi \cos(\phi)}, \cdots, e^{\mathrm{j}\pi (N-1) \cos(\phi)}]^\mathrm{T}$, where $N$ is the number of antennas and $\phi \in \{\phi_{m_t,q}, \phi_{m_r,q}\}$. Adopting the Swerling-I model, the channel gain remains constant over $L$ symbols. In practice, targets may not always be present. We model the presence of the $q$-th target using a Bernoulli random variable $\alpha_q \in \{0,1\}$ with probability $\mathbb{P}(\alpha_q=1)=p_q$. Accordingly, the received signal at the receiving AP is expressed as
\begin{equation}
\mathbf{y}_{m_r}^{(\mathrm{s})}[\ell] = \sum_{q=1}^{Q} \alpha_q \sum_{m_t \in \mathcal{M}_t} \mathbf{G}_{m_t, q, m_r} \mathbf{x}_{m_t}[\ell] + \mathbf{w}_{m_r}[\ell],
\end{equation}
where $\mathbf{w}_{m_r}[\ell] = \mathbf{c}_{m_r}[\ell] + \mathbf{n}_{m_r}[\ell]$ denotes the aggregate disturbance consisting of Gaussian clutter $\mathbf{c}_{m_r}[\ell] \sim \mathcal{CN}(\mathbf{0}, \sigma_{\mathrm{c}, m_r}^2 \mathbf{I})$ and receiver noise $\mathbf{n}_{m_r}[\ell] \sim \mathcal{CN}(\mathbf{0}, \sigma_{\mathrm{n}, m_r}^2 \mathbf{I})$. The aggregate disturbance $\mathbf{w}_{m_r}[\ell]$ follows a complex Gaussian distribution $\mathcal{CN}(\mathbf{0}, \sigma_{\mathrm{w}, m_r}^2 \mathbf{I})$, with the total variance given by $\sigma_{\mathrm{w}, m_r}^2 = \sigma_{\mathrm{c}, m_r}^2 + \sigma_{\mathrm{n}, m_r}^2$~\cite{b16}. To express the radar received signal compactly, we introduce $\overline{\mathbf{V}}_m = [\mathbf{v}_{m1}, \ldots, \mathbf{v}_{mK}] \in \mathbb{C}^{N \times K}$, $\mathbf{x}_k= \begin{bmatrix} x_k[1],\ldots,x_k[L] \end{bmatrix}^\mathrm{T}$ and $\overline{\mathbf{X}} = [\mathbf{x}_1, \ldots, \mathbf{x}_K]^\mathrm{T} \in \mathbb{C}^{K \times L}$. Consequently, the signal transmitted by AP $m_t$ over $L$ symbols can be represented as $\mathbf{X}_{m_t} = \overline{\mathbf{V}}_{m_t} \overline{\mathbf{X}} \in \mathbb{C}^{N \times L}$. Therefore, the received signal due to the $L$ symbols at the $m_r$-th receiving AP can be expressed as

\begin{equation}
\mathbf{Y}_{m_r}^{(\mathrm{s})} = \sum_{q=1}^{Q} \alpha_q \sum_{m_t \in \mathcal{M}_t} \mathbf{G}_{m_t, q, m_r} \mathbf{X}_{m_t} + \mathbf{W}_{m_r},
\end{equation}
where $\mathbf{W}_{m_r} = \left[\mathbf{w}_{m_r}[1], \ldots, \mathbf{w}_{m_r}[L]\right]$ denotes the aggregate disturbance matrix.
\begin{algorithm}[t]
\caption{SDR-Based Multi-AP Cooperative Beamforming Algorithm}\label{alg:SDR}
\begin{algorithmic}
\STATE \textbf{INPUT:} $\mathcal{M}_t, \mathcal{M}_r, M, N, K, Q, P_m, R_k^{\min}$
\STATE \textbf{OUTPUT:} $\{\mathbf{v}_k\}$
\STATE \textbf{Step 1: Initialization}
\STATE \hspace{0.5cm} 1.1. Compute $\Gamma_k^{\min} = 2^{R_k^{\min}} - 1$
\STATE \hspace{0.5cm} 1.2. Define $\mathbf{h}_k$ by aggregating $\mathbf{h}_{mk}$ for all APs
\STATE \hspace{0.5cm} 1.3. Obtain $\gamma^{(\mathrm{s})}$ using equation \eqref{eq:q13}
\STATE \textbf{Step 2: Formulate the Optimization Problem}
\STATE \hspace{0.5cm} 2.1. Construct optimization problem \textbf{P2}
\STATE \hspace{0.5cm} 2.2. Solve the SDR of \textbf{P2} using CVX
\STATE \textbf{Step 3: Beamforming Vector Extraction}
\STATE \hspace{0.5cm} 3.1. Obtain $\mathbf{v}_k^{\prime\prime}$ using equation \eqref{eq:vk_double_prime}
\STATE \textbf{Step 4: Output Results}
\STATE \hspace{0.5cm} 4.1. \textbf{return} $\{\mathbf{v}_k\}$
\end{algorithmic}
\end{algorithm}

In detection tasks, the sensing SCNR is a critical physical-layer metric in radar systems, as it directly influences the ability to distinguish targets from background noise. The expected sensing SCNR is derived as $\gamma^{(\mathrm{s})}$ in (\ref{eq:q9}) (see the bottom of this page). It should be noted that the SCNR considered here represents a joint radar SCNR metric.

\begin{figure*}[b] 
\hrulefill
\begin{equation}
\begin{aligned}
\gamma^{(\mathrm{s})} = \frac{\mathbb{E}\left[\sum_{m_r\in\mathcal{M}_r}\left\|\sum_{q=1}^Q\alpha_q\sum_{m_t\in\mathcal{M}_t}\mathbf{G}_{m_t, q, m_r} \mathbf{X}_{m_t}\right\|_\mathrm{F}^2\right]}{\mathbb{E}\left[\sum_{m_r\in\mathcal{M}_r}\left\|\mathbf{W}_{m_r}\right\|_\mathrm{F}^2\right]}= \frac{\sum_{m_r\in\mathcal{M}_r}\sum_{m_t\in\mathcal{M}_t}\sum_{q=1}^Q p_q\zeta_{m_t, q, m_r}^2\left\|\mathbf{a}^\mathrm{H}(\phi_{m_t, q})\overline{\mathbf{V}}_{m_t}\right\|^2}{\sum_{m_r\in\mathcal{M}_r}\sigma_{\mathrm{w},m_r}^2}.\label{eq:q9}
\end{aligned}
\end{equation}
\end{figure*}

\section{Cooperative Beamforming for CF-ISAC}
In this section, we design the beamforming vectors for the transmitting APs in cell-free ISAC networks. Our objective is to maximize the sensing SCNR while satisfying the constraints on the communication rate threshold and the AP transmit power. This optimization problem can be formulated as
\begin{subequations}\label{eq:P1}
\begin{align}
(\mathrm{P}1):\quad & \max_{\{\mathbf{v}_{mk}\}} \quad \gamma^{(\mathrm{s})} \label{eq:P1a} \\
& \text{s.t.} \quad R_k \ge R_k^{\min}, \quad \forall k, \label{eq:P1b} \\
& \qquad \sum_k \|\mathbf{v}_{mk}\|^2 \le P_m, \quad \forall m, \label{eq:P1c}
\end{align}
\end{subequations}
where $R_k^{\min}$ is the minimum required communication rate for user $k$. The first constraint \eqref{eq:P1b} ensures that the communication rate of each user is above the minimum required rate, thus maintaining a certain level of communication quality. The second constraint \eqref{eq:P1c} guarantees that the total power used by the $m$-th AP does not exceed its maximum available power. 

Since the objective \eqref{eq:P1a} is non-convex, problem \textbf{P1} is a non-convex quadratically constrained quadratic program (QCQP) and cannot be solved directly. To address this, we reformulate \textbf{P1} into an equivalent semidefinite programming (SDP) form and apply the SDR technique~\cite{b11}. The resulting convex problem is then efficiently solved using standard solvers. The detailed implementation steps are as follows.

We need to define some new variables. First, we define the access point selection matrix $\mathbf{A}_m \in \mathbb{R}^{MN \times MN}$ as follows:
\begin{equation}
\left[\mathbf{A}_m\right]_{ij} = 
\begin{cases} 
1 & \text{if } (m-1)N+1 \leq i \leq mN, i = j, \\ 
0 & \text{otherwise}. 
\end{cases}
\end{equation}
Additionally, we reformulate the beamforming variables as $\mathbf{V}_k = \mathbf{v}_k \mathbf{v}_k^{\mathrm{H}} \in \mathbb{S}^+, \ \forall k \in \mathcal{K}$, with the non-convex rank-one constraint $\operatorname{rank}(\mathbf{V}_k) = 1$. Let $\overline{\mathbf{V}} = \sum_{k \in \mathcal{K}} \mathbf{V}_k$ denote the aggregate variable. To simplify the sensing SCNR,
we define $\mathbf{E} = \sum_{m_t \in \mathcal{M}_t}\sum_{m_r \in \mathcal{M}_r}\sum_{q=1}^Qp_q \zeta_{m_t, q, m_r}^2 \mathbf{A}_{m_t} \overline{\mathbf{E}}_q \mathbf{A}_{m_t}^\mathrm{H}$, where $\overline{\mathbf{E}}_q = \overline{\mathbf{a}}_q \overline{\mathbf{a}}_q^\mathrm{H}$ and $\overline{\mathbf{a}}_q = [\mathbf{a}(\phi_{1, q})^\mathrm{T}, \dots, \mathbf{a}(\phi_{M_t, q})^\mathrm{T}]^\mathrm{T}$. Based on these definitions, the SCNR numerator in \eqref{eq:q9} is derived as
\begin{equation}
\begin{aligned}
&\mathbb{E}\left[ \sum_{m_r \in \mathcal{M}_r} \left\| \sum_{q=1}^Q \alpha_q \sum_{m_t \in \mathcal{M}_t} \mathbf{G}_{m_t, q, m_r} \mathbf{X}_{m_t} \right\|_F^2 \right] \\
&= NL \sum_{q=1}^Q  \sum_{m_r \in \mathcal{M}_r}p_q \zeta_{m_t, q, m_r}^2  \operatorname{Tr}\left( \sum_{m_t \in \mathcal{M}_t} \mathbf{A}_{m_t} \overline{\mathbf{E}}_q \mathbf{A}_{m_t}^\mathrm{H} \overline{\mathbf{V}} \right) \\
&= NL \operatorname{Tr}\left( \mathbf{E} \overline{\mathbf{V}} \right),
\end{aligned}
\end{equation}
Then, the sensing SCNR can be defined as
\begin{equation}\label{eq:q13}
\gamma^{(\mathrm{s})} = \frac{\operatorname{Tr}(\mathbf{E} \overline{\mathbf{V}})}{\sum_{m_r \in \mathcal{M}_r} \sigma_{\mathrm{w}, m_r}^2}.
\end{equation}

According to the constraint conditions in (\ref{eq:P1b}), we define $\Gamma_k^{\min} = 2^{R_k^{\min}} - 1$, from which it follows that the communication SINR requirement is equivalently expressed as $\gamma_k^{(\mathrm{c})} \ge \Gamma_k^{\min}$. Let $\mathbf{H}_k = \mathbf{h}_k \mathbf{h}_k^{\mathrm{H}}$, and rewrite the SINR expression in (\ref{eq:4}) in terms of the new variable as
\begin{equation}
\gamma_k^{(\mathrm{c})} = 
\frac{\operatorname{Tr}\left(\mathbf{H}_k \mathbf{V}_k\right)}
{\sum_{k^\prime \in \mathcal{K} \setminus \{k\}} \operatorname{Tr}\left(\mathbf{H}_k \mathbf{V}_{k^\prime}\right) + \sigma_k^2}.
\end{equation}

Then, the above communication rate constraint in \eqref{eq:P1b} can be equivalently formulated as
\begin{equation}
\begin{aligned}
(1 + \Gamma_k^{\min}) \operatorname{Tr}\!\left(\mathbf{H}_k \mathbf{V}_k\right)
- \operatorname{Tr}\!\left(\mathbf{H}_k \overline{\mathbf{V}}\right)
\ge \sigma_k^{2}.\label{eq:q15}
\end{aligned}
\end{equation}
Based on the previous derivations, we now present the SDP formulation \textbf{P2} of the optimization problem.
\begin{subequations}\label{eq:P2}
\begin{align}
(\mathrm{P}2):\quad 
& \max_{\{\mathbf{V}_k\},\, \overline{\mathbf{V}}} 
\quad \operatorname{Tr}\!\left( \mathbf{E} \overline{\mathbf{V}} \right) 
\label{eq:P2a} \\[2pt]
& \text{s.t.} \quad \hspace*{\fill} (\ref{eq:q15}) \hspace*{\fill} 
\label{eq:P2b} \\
& \qquad \operatorname{Tr}\!\left( \mathbf{A}_m \overline{\mathbf{V}} \right) 
\le P_m, \quad \forall m \in \mathcal{M}_t, 
\label{eq:P2c} \\
& \qquad \mathbf{V}_k \in \mathbb{S}^{+}, 
\quad \forall k \in \mathcal{K}, 
\label{eq:P2d} \\
& \qquad \operatorname{rank}\!\left( \mathbf{V}_k \right) = 1, 
\quad \forall k \in \mathcal{K}. 
\label{eq:P2e}
\end{align}
\end{subequations}

By relaxing the non-convex rank constraints in \eqref{eq:P2e}, problem \textbf{P2} is transformed into a convex SDP. This relaxed problem can be efficiently solved using interior-point methods, which guarantee convergence to a globally optimal solution in polynomial time. However, to rigorously guarantee the global optimality of the original non-convex problem \textbf{P1}, we must establish the tightness condition of the SDR and the existence of a rank-one solution. 

\begin{proposition}[SDR Tightness]\label{prop:tightness}
The SDR of \textbf{P2} is tight. A globally optimal rank-one solution is guaranteed by the rank-one structure $\mathbf{H}_k = \mathbf{h}_k\mathbf{h}_k^{\mathrm{H}}, \forall k \in \mathcal{K}$.
\end{proposition}

Specifically, let $\{\mathbf{V}_k^\prime\}$ be an optimal solution to the relaxed SDP. We can reconstruct an exact rank-one optimal solution $\mathbf{V}_k^{\prime\prime} = \mathbf{v}_k^{\prime\prime}(\mathbf{v}_k^{\prime\prime})^{\mathrm{H}}$ by defining
\begin{equation}\label{eq:vk_double_prime}
\mathbf{v}_k^{\prime\prime} = (\mathbf{h}_k^{\mathrm{H}} \mathbf{V}_k^\prime \mathbf{h}_k )^{-\frac{1}{2}} \mathbf{V}_k^\prime \mathbf{h}_k.
\end{equation}
\begin{IEEEproof}
Let $\overline{\mathbf{V}}^\prime \!=\! \sum_{k} \mathbf{V}_k^\prime$. We construct a new solution $\{\mathbf{V}_k^{\prime\prime}\}$ using \eqref{eq:vk_double_prime}, and let $\overline{\mathbf{V}}^{\prime\prime} \!=\! \overline{\mathbf{V}}^\prime$. By definition, $\mathbf{V}_k^{\prime\prime} \succeq \mathbf{0}$ and $\mathrm{rank}(\mathbf{V}_k^{\prime\prime})\!=\!1$, which satisfy \eqref{eq:P2d} and \eqref{eq:P2e}. Since $\overline{\mathbf{V}}^{\prime\prime} \!=\! \overline{\mathbf{V}}^\prime$, both the objective value in \eqref{eq:P2a} with $\mathrm{Tr}(\mathbf{E}\overline{\mathbf{V}}^{\prime\prime}) \!=\! \mathrm{Tr}(\mathbf{E}\overline{\mathbf{V}}^\prime)$ and the power constraints in \eqref{eq:P2c} with $\mathrm{Tr}(\mathbf{A}_m\overline{\mathbf{V}}^{\prime\prime}) \!=\! \mathrm{Tr}(\mathbf{A}_m\overline{\mathbf{V}}^\prime) \!\le\! P_m$ remains perfectly unchanged. For the rate constraints in \eqref{eq:P2b}, the interference terms depend only on $\overline{\mathbf{V}}^{\prime\prime}$ and thus remain identical. It suffices to verify the signal power term. Based on \eqref{eq:vk_double_prime}, we obtain $\mathrm{Tr}(\mathbf{H}_k\mathbf{V}_k^{\prime\prime}) = | (\mathbf{h}_k^\mathrm{H}\mathbf{V}_k^\prime\mathbf{h}_k)^{-1/2} \mathbf{h}_k^\mathrm{H}\mathbf{V}_k^\prime\mathbf{h}_k |^2 \!=\! \mathbf{h}_k^\mathrm{H}\mathbf{V}_k^\prime\mathbf{h}_k = \mathrm{Tr}(\mathbf{H}_k\mathbf{V}_k^\prime)$. Thus, $\{\mathbf{V}_k^{\prime\prime}\}$ satisfies all constraints of \eqref{eq:P1} without objective degradation, proving the SDR is tight.
\end{IEEEproof}

In summary, Algorithm \ref{alg:SDR} outlines the procedure for our proposed SDR-based multi-AP cooperative beamforming (SDR-MCBF) algorithm.

In the centralized architecture, the CPU distributes the coordinated beamforming vectors to the $M_t$ transmitting APs. For each coherence interval, the CPU transmits $N$ complex coefficients to each AP per user stream. Consequently, the total fronthaul overhead is $M_tKN$ complex scalars. Although higher than that of distributed schemes, this overhead is necessary to achieve the coherent beamforming gain required to maximize the sensing SCNR.
\section{Simulation Results}

In this section, simulation results are presented to validate the effectiveness of our proposed methodology. We consider a $200 \times 200$ $\text{m}^2$ CF-ISAC scenario. Two 16-antenna  transmitting APs and two receiving APs are deployed at $(\pm 60, 60)$ m and $(\pm 60, -60)$~m, respectively. We distribute $K=5$ users uniformly along $y=20$ m with $x \in [-100, 100]$ m. Furthermore, $Q=2$ potential targets appear around $(0, 20)$ m with probability $p_q=0.5$, with their exact locations generated randomly. The communication links undergo block flat Rayleigh fading with path loss $\zeta_{\mathbf{h}_{mk}} = C_0 (d_{mk}/D_0)^{-\nu}$ ($C_0 = -30$ dB, $\nu = 3.0$). For the sensing channels, we adopt
the parameters $\zeta_{m_t, q, m_r}^2=0.1$ and $\sigma_{\mathrm{w}, m_r}^2=1$. With AP transmit power $P_m = 0$ dBW, user noise $\sigma_k^2=-80$ dBm, all simulation results are averaged over $1000$ Monte Carlo realizations to ensure statistical robustness.
\begin{figure}[htbp]
    \centering
    \includegraphics[width=0.5\textwidth, height=0.27\textwidth, trim=30 10 30 10, clip]{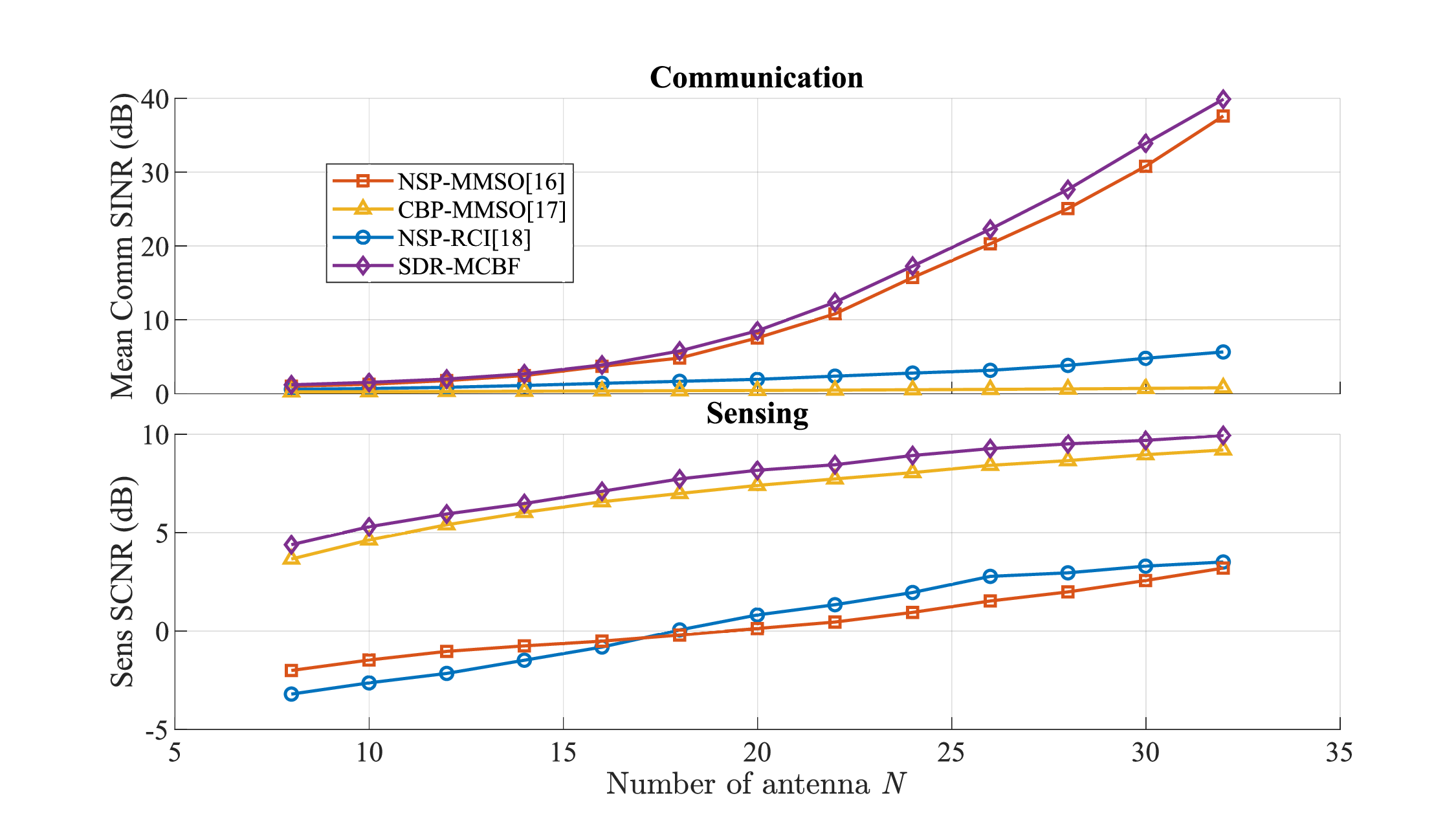}
    
    \caption{The communication and sensing performance of various schemes under different antenna numbers $N$.}
    \label{fig:2}
\end{figure}
\begin{figure}[htbp]
    \centering
    \includegraphics[width=0.60\textwidth, height=0.27\textwidth, trim=0 0 0 0, clip]{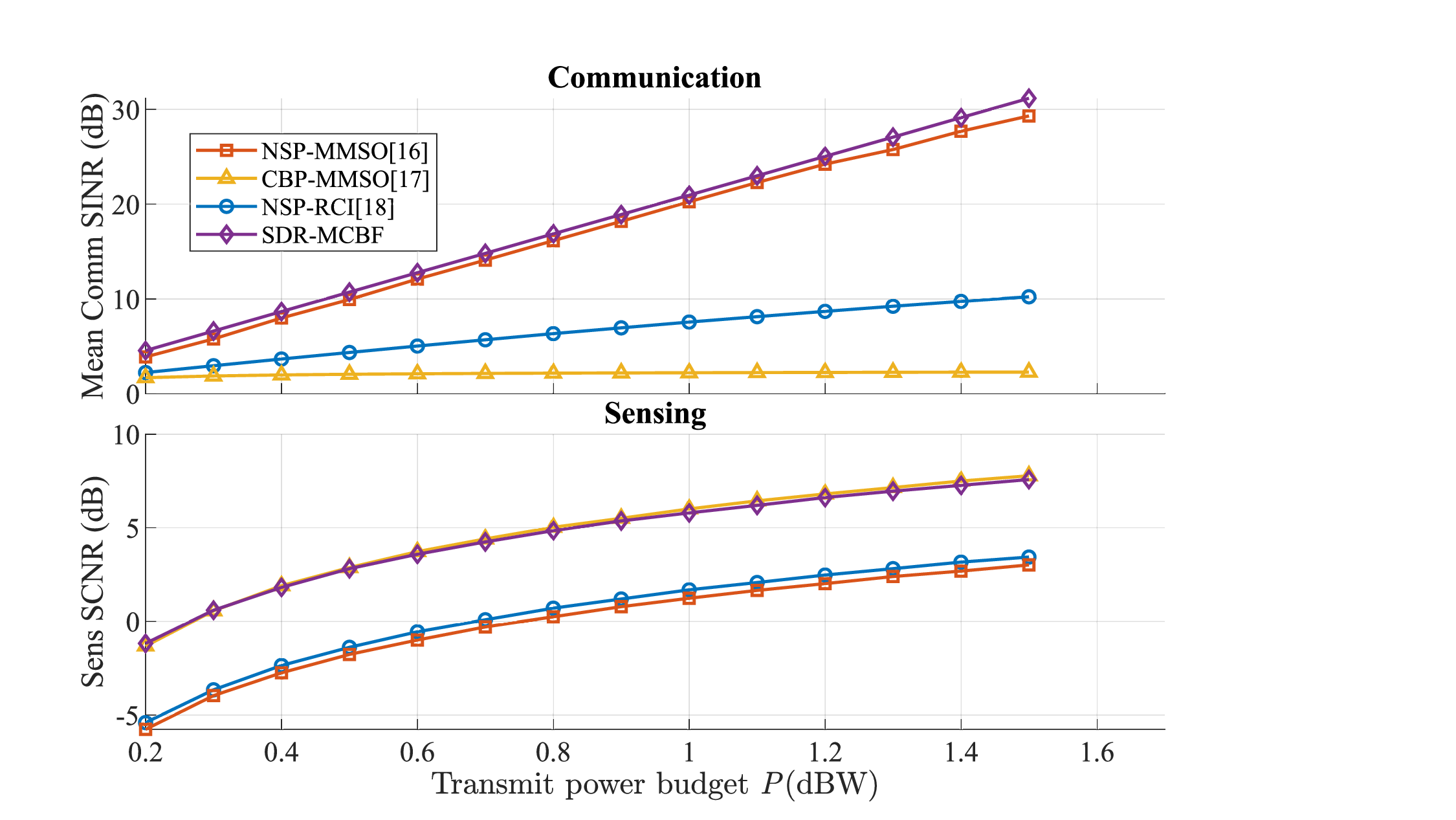}
    
    \caption{The communication and sensing performance of various schemes under different power $P$.}
    \label{fig:3}
\end{figure}
\begin{figure}[htbp]
    \centering
    \includegraphics[width=0.55\textwidth, height=0.29\textwidth, trim=30 10 30 10, clip]{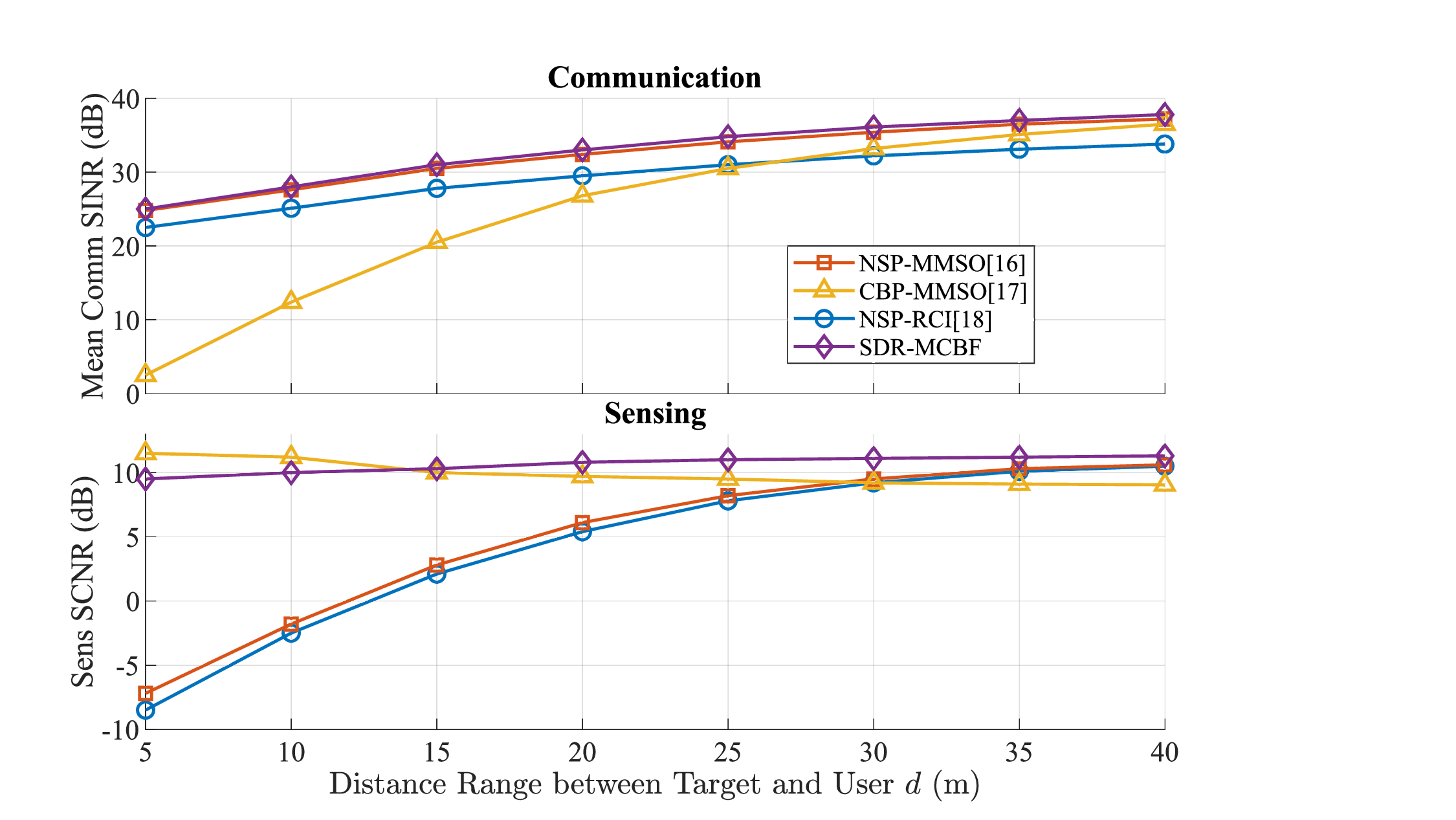}
    
    \caption{Communication and sensing performance versus the distance between sensing targets and communication users $d$.}
    \label{fig:4}
\end{figure}

To analyze the performance differences, we construct three benchmark schemes by linearly combining dedicated sensing and communication beamformers via an optimal weighting coefficient $\lambda \in [0, 1]$. The base methods include Null-Space Projection (NSP), Conjugate Beamforming Precoding (CBP), Regularized Channel Inversion (RCI), and Max-Min SINR Optimal (MMSO), denoted as $\mathbf{v}_{mk}^{\mathrm{NSP}}$, $\mathbf{v}_{mk}^{\mathrm{CBP}}$, $\mathbf{v}_{mk}^{\mathrm{RCI}}$, and $\mathbf{v}_{mk}^{\mathrm{MMSO}}$, respectively. Based on the linear combination framework $\mathbf{v}_{mk}^{\mathrm{Baseline}} = \lambda \mathbf{v}_{mk}^{\mathrm{Sens}} + (1-\lambda) \mathbf{v}_{mk}^{\mathrm{Comm}}$, we establish three benchmarks: NSP-MMSO~\cite{b13}, CBP-MMSO~\cite{b14}, and NSP-RCI~\cite{b15}, by substituting the corresponding sensing and communication beamformer pairs.

As illustrated in Fig.~\ref{fig:2}, increasing the number of antennas $N$ per AP drives a continuous upward trend in both communication SINR and sensing SCNR. Expanding the array scale significantly enhances beamforming capabilities and increases the number of spatial degrees of freedom. A similar overall performance improvement is observed when the transmit power is increased, as shown in Fig.~\ref{fig:3}. However, the sensing SCNR improves rapidly before saturating, while the communication SINR grows linearly. Across these varying physical parameters, the baseline schemes reveal inherent trade-offs in performance. CBP-MMSO achieves near-optimal sensing at the severe expense of communication. NSP-based methods (NSP-RCI and NSP-MMSO) yield limited sensing gains due to strict null-space constraints. In contrast, the proposed SDR-MCBF approach consistently achieves the near-optimal overall performance. By effectively coordinating resources and mitigating interference, SDR-MCBF excels in both power-limited and large-array scenarios. It demonstrates superior adaptability, scalability, and a balanced enhancement across both dual-functional objectives.

Fig.~\ref{fig:4} shows that under severe spatial conflict, with the target and user in close proximity, NSP-RCI and NSP-MMSO experience SCNR drops in sensing, while CBP-MMSO
experiences severe communication SINR degradation. In contrast, the proposed SDR-MCBF algorithm exhibits exceptional spatial resilience, achieving an optimal balance by maintaining robust communication and high sensing performance. Furthermore, SDR-MCBF consistently maintains its overall dual-functional superiority as the spatial separation increases.

\section{Conclusion}
This letter investigates the joint optimization problem of communication and sensing performance in cell-free integrated sensing and communication networks with distributed multi-AP cooperation. The SDR-based Multi-AP cooperative beamforming approach is proposed to maximize the sensing SCNR while satisfying communication rate constraints. Simulation results verify that the proposed approach achieves superior overall performance, demonstrating its effectiveness, adaptability, and scalability in complex network environments.

\vfill


\begin{thebibliography}{99}
\bibliographystyle{IEEEtran}

\bibitem{b1} W. Xu, Z. Yang, D. W. K. Ng, M. Levorato, Y. C. Eldar, and M. Debbah, ``Edge learning for B5G networks with distributed signal processing: Semantic communication, edge computing, and wireless sensing,'' \textit{IEEE J. Sel. Topics Signal Process.}, vol. 17, no. 1, pp. 9--39, Jan. 2023.

\bibitem{b2} W. Shi, W. Xu, X. You, C. Zhao, and K. Wei, ``Intelligent reflection enabling technologies for integrated and green Internet-of-Everything beyond 5G: Communication, sensing, and security,'' \textit{IEEE Wireless Commun.}, vol. 30, no. 2, pp. 147--154, Apr. 2023.

\bibitem{b19}
W. Liang, Y. Wang, J. Zhang, L. Li, and Z. Han,
``Predictive beamforming in integrated sensing and communication-enabled vehicular networks,''
\textit{IEEE Trans. Veh. Technol.}, vol. 74, no. 3, pp. 4539--4553, Mar. 2025.

\bibitem{b3} W. Mao, Y. Lu, C.-Y. Chi, B. Ai, Z. Zhong, and Z. Ding, ``Communication-sensing region for cell-free massive MIMO ISAC systems,'' \textit{IEEE Trans. Wireless Commun.}, vol. 23, no. 9, pp. 12396--12411, Sep. 2024.

\bibitem{b4} N. Zhao, Y. Wang, Z. Zhang, Q. Chang, and Y. Shen, ``Joint transmit and receive beamforming design for integrated sensing and communication,'' \textit{IEEE Commun. Lett.}, vol. 26, no. 3, pp. 662--666, Mar. 2022.

\bibitem{b6} Z. Behdad, Ö. T. Demir, K. W. Sung, E. Björnson, and C. Cavdar, ``Multi-static target detection and power allocation for integrated sensing and communication in cell-free massive MIMO,'' \textit{IEEE Trans. Wireless Commun.}, vol. 23, no. 9, pp. 11580--11596, Sep. 2024.

\bibitem{b7} Y. Zhang, H. Shan, Y. Zhou, Z. Shi, L. Sheng, and Y. Liu, ``Cooperative beamforming design for anti-UAV ISAC systems,'' \textit{IEEE Trans. Wireless Commun.}, vol. 24, no. 3, pp. 2249--2264, Mar. 2025.

\bibitem{b20}
Y. Zheng, L. Li, W. Lin, W. Liang, Q. Du, and Z. Han,
``Optimal transport framework for ISAC in low-altitude networks: Joint resource allocation for cooperative communication and non-cooperative localization,''
\textit{IEEE Trans. Commun.}, vol. 74, pp. 1984--2000, Mar. 2026.

\bibitem{b8} X. Yang, Z. Wei, J. Xu, Y. Fang, H. Wu, and Z. Feng, ``Coordinated transmit beamforming for networked ISAC with imperfect CSI and time synchronization,'' \textit{IEEE Trans. Wireless Commun.}, vol. 23, no. 12, pp. 18019--18035, Dec. 2024.

\bibitem{b17} A. Chowdhury, S. S. Thoota, and E. G. Larsson, ``Joint sensing and bi-directional communication with dynamic TDD enabled cell-free MIMO,'' \textit{IEEE Trans. Wireless Commun.}, vol. 25, pp. 2420--2434, 2026.

\bibitem{b18} H. Liu, Z. Wei, L. Sun, R. Xu, Y. Zhang, and Z. Feng, ``Cooperative sensing in cell-free massive MIMO ISAC systems: Performance optimization and signal processing,'' \textit{IEEE Trans. Wireless Commun.}, vol. 25, pp. 12531--12547, 2026.

\bibitem{b12} S. Zargari, D. Galappaththige, C. Tellambura, and H. V. Poor, ``A Riemannian manifold approach to constrained resource allocation in ISAC,'' \textit{IEEE Trans. Commun.}, vol. 73, no. 5, pp. 3655--3670, May 2025.

\bibitem{b10} Z. Behdad, Ö. T. Demir, K. W. Sung, E. Björnson, and C. Cavdar, ``Power allocation for joint communication and sensing in cell-free massive MIMO,'' in \textit{Proc. IEEE Global Commun. Conf. (GLOBECOM)}, Rio de Janeiro, Brazil, Dec. 2022, pp. 4081--4086.

\bibitem{b16} Z. Wei \textit{et al.}, ``Integrated sensing and communication channel modeling: A survey,'' \textit{IEEE Internet Things J.}, vol. 12, no. 12, pp. 18850--18864, Jun. 2025.

\bibitem{b11} M. Grant and S. Boyd, ``CVX: Matlab software for disciplined convex programming, version 2.1,'' Mar. 2014. [Online]. Available: http://cvxr.com/cvx

\bibitem{b13} U. Demirhan and A. Alkhateeb, ``Cell-free ISAC MIMO systems: Joint sensing and communication beamforming,'' \textit{IEEE Trans. Commun.}, vol. 73, no. 6, pp. 4454--4468, Jun. 2025.

\bibitem{b14} A. Zhou, J. Wu, E. G. Larsson, and P. Fan, ``Max-min optimal beamforming for cell-free massive MIMO,'' \textit{IEEE Commun. Lett.}, vol. 24, no. 10, pp. 2344--2348, Oct. 2020.

\bibitem{b15} J. Suh, J. Kang, K. Han, S. Hong, and G.-T. Gil, ``Null space projection-based design of multibeam for joint communication and sensing systems,'' \textit{IEEE Commun. Lett.}, vol. 27, no. 8, pp. 2162--2166, Aug. 2023.

\end{thebibliography}
\end{document}